\documentclass[letterpaper]{article} 
\usepackage{aaai24}  
\usepackage{times}  
\usepackage{helvet}  
\usepackage{courier}  
\usepackage[hyphens]{url}  
\usepackage{graphicx} 
\usepackage{natbib}  
\usepackage{caption} 
\usepackage{todonotes}
\usepackage{algorithmic}
\usepackage[ruled,linesnumbered,vlined]{algorithm2e}
\usepackage[T1]{fontenc}
\usepackage{graphicx}
\usepackage{times}
\usepackage{soul}
\usepackage{url}
\usepackage[utf8]{inputenc}
\usepackage{amsmath}
\usepackage{amsthm}
\usepackage{amssymb}
\usepackage{xcolor}
\usepackage{enumitem}
\usepackage{xspace}
\usepackage{lineno}

\newtheorem{theorem}{Theorem}
\newtheorem{observation}{Observation}
\newtheorem{definition}{Definition}
\newtheorem{proposition}{Proposition}
\newtheorem{lemma}{Lemma}
\newtheorem{cclaim}{Claim}
\newtheorem{corollary}{Corollary}
\newtheorem{property}{Property}
\newtheorem{hypothesis}{Hypothesis}

\newcommand{\sv}[1]{#1}
\newcommand{\lv}[1]{}
\newcommand{\sba}{$\mathsf{SBA}$\xspace}
\newcommand{\pba}{$\mathsf{PBA}$\xspace}
\newcommand{\ba}{$\mathsf{BA}$\xspace}
\newcommand{\npc}{NP--complete}

\newcommand{\descendant}{{\sf Desc}}
\newcommand{\fdesc}{{\sf FDesc}}

\newcommand{\child}{{\sf Child}}
\newcommand{\sibling}{{\sf Sibl}}
\newcommand{\bigoh}[0]{{\mathcal{O}}}
\newcommand{\head}{{\sf Lose}}
\newcommand{\troot}{{\sf root}}
\newcommand{\cS}{\mathcal{S}}

\renewcommand{\alpha}{\textsf{ht}}
\renewcommand{\beta}{\textsf{sz}}
\newcommand{\demandTF}{Demand-TF\xspace}

\newcommand{\NPH}{\textsf{\textup{NP-hard}}\xspace}
\newcommand{\TFP}{{\sc TF}\xspace}
\newcommand{\Scal}[0]{\mathcal{S}}

\title{An Exercise in Tournament Design: When Some Matches Must Be Scheduled}

\author {
Sushmita Gupta\textsuperscript{\rm 1},
	Ramanujan Sridharan\textsuperscript{\rm 2},
 	Peter Strulo\textsuperscript{\rm 2}	
}

\affiliations {    \textsuperscript{\rm 1}The Institute of Mathematical Sciences, HBNI, India\\
    \textsuperscript{\rm 2}University of Warwick, UK\\
sushmita.gupta@gmail.com, r.maadapuzhi-sridharan@warwick.ac.uk,  peter.strulo@warwick.ac.uk
}

\begin{document}

\maketitle

\begin{abstract}
Single-elimination (SE) tournaments are a popular format used in competitive environments and decision making. Algorithms for SE tournament manipulation  have been an active topic of research in recent years. In this paper, we initiate the algorithmic study of a novel variant of SE tournament manipulation that aims to model the fact that certain matchups are highly desired in a sporting context, incentivizing an organizer to manipulate the bracket to make such matchups take place. 
We obtain both hardness and tractability results. We show that while the problem of computing a bracket enforcing a given set of matches in an SE tournament is NP-hard, there are natural restrictions that lead to polynomial-time solvability. In particular, we show polynomial-time solvability if there is a linear ordering on the ability of players with only a constant number of exceptions where a player with lower ability beats a player with higher ability.
\end{abstract}

\section{Introduction}\label{sec:intro}

There is a rich history of work on the algorithmics of designing Single Elimination (SE) or knockout tournaments as they are a format of competition employed in varied scenarios such as sports, elections and different forms of decision making~\cite{Tullock80,DBLP:journals/ior/HorenR85,Rosen86,Laslier97,CR11}. 
Based on an initial bracket (a permutation of the players, also called a {\em seeding}), it proceeds in multiple rounds, culminating in a single winner. In each round, all players that have not yet lost a match are paired up to play the next set of matches. Losers exit and winners proceed to the next round, until only one remains, the winner of the tournament. 
In general, the tournament designer is assumed to be given probabilities $p_{ij}$ expressing the likelihood that player $i$ beats player $j$. In this paper, we focus on the deterministic model, i.e., when these probabilities are 0 or 1. This model has already been the subject of numerous papers in the last few years. Besides being independently interesting from a structural and algorithmic perspective as shown by \cite{Williams10,AzizGMMSW14,RamanujanS17,GuptaR0Z18a,Gupta0SZ19,ManurangsiSuksompong,Zehavi23}, the deterministic model naturally captures sequential majority elections along binary trees~\cite{Lang07,VuAS2009} where each ``match'' is a  comparison of votes of two candidates and the candidate with more votes wins and moves on.

A major question in the study of SE tournament design is the {\sc  Tournament Fixing} Problem (TF): Can a designer efficiently find a bracket that maximizes the likelihood (or ensures, in the case of the deterministic model) that a player of their choice wins the tournament?
However, there are other natural objectives around SE tournament design besides favoring a particular player and that is the focus of this paper.  Great rivalries generate great entertainment. Imagine a sports tournament that features marquee matches marked by factors such as historic rivalries, contemporaneous news events, geographic proximity or even personal rivalries between members of the opposing teams. 
These are some of the most widely known and talked about rivalries in the world of sports that greatly enhanced the notoriety and visibility of the sport and thereby achieved great financial success, publicity and relevancy for all the stake holders, be it the organizers, the sponsors, not to mention the participants. 
From the competition design perspective, that considerations such as revenue, viewer engagement, and relevance should be at the forefront is straightforward.

Thus, scheduling especially attractive matches is a rational tournament design imperative, motivating  
our algorithmic study of finding a bracket that aims to ensure that a given set of {\it demand matches} are played in the  SE tournament.
One can view our model of scheduling a set of demanded games in this setting 
to be a special case of revenue maximization with unit revenue given to each demand match and zero to all others. Setting the target revenue to be equal to the number of demand matches implies that achieving the target revenue is the same as scheduling all the demand matches. A more general problem was studied by \cite{Lang07}, in the context of sequential majority voting along binary trees, except they allow arbitrary costs for the edges and the goal is to achieve minimum possible cost in the final SE tournament. They showed this problem to be NP-hard. We also note that our paper is naturally aligned with investigations into the relationship between round-robin and SE tournaments, e.g.,~\cite{DBLP:conf/ijcai/StantonW11}. Specifically, while every demand match obviously occurs in a round-robin tournament,  how efficiently could one ensure the same in an SE tournament? This is our focus.

\paragraph{Our contributions.} We make advances on two fronts-- conceptual and algorithmic. On the conceptual front, we introduce a tournament design objective that is both a special case of revenue maximization and as we demonstrate later, a novel variant of the well-studied Subgraph Isomorphism problem. We call the problem {\sc Demand Tournament Fixing} (Demand-TF), formally defined as follows. 

The input contains (i) a directed graph $T$ (called tournament digraph) whose vertices are the players and for {\em every} pair of players $u$ and $v$, there is a directed edge (i.e., arc) from $u$ to $v$ if and only if $u$ beats $v$ (assume no ties) and (ii) a set 
$\Scal$ of arcs of $T$. The vertices of $T$ are denoted by $V(T)$ and the arcs by $A(T)$. The goal is to find a bracket (if one exists) such that in the SE tournament generated by this bracket, the matches corresponding to the arcs in $\Scal$ (called demand matches) take place. 

Solving this problem requires one to create a bracket that will ensure that a player $u$ (and its ``demand rivals'') {\em all} progress far enough in the tournament so that $u$ is able to play in all the demand matches featuring it.  
Effectively, we are aiming to create within one single bracket, multiple favorable brackets for each of those players that are somehow highly correlated. Treading this fine line raises fascinating algorithmic challenges as we show in this paper. In fact, for the special case of an acyclic tournament digraph (DAG), i.e., when there is a linear ordering of players according to their strengths where each player beats every player appearing after it, the TF problem is trivial as the strongest player wins every SE tournament. On the other hand, even in this special case, {\demandTF} is a challenging problem. However, as we will discuss, our main result implies a polynomial-time algorithm even for this problem on DAGs, as a corollary. 

We highlight the relation between our problem and {\sc Subgraph Isomorphism}~\cite{CyganFKLMPPS15} problem (SI). 
\TFP has a well-established connection with a specific type of spanning tree within the tournament digraph, called a {\it spanning binomial arborescence} or {\sba}~\cite{DBLP:reference/choice/Williams16}. We refer the reader to the section on preliminaries for a formal definition. In fact, there is a solution to the \TFP instance if and only if the tournament digraph has an {\sba} rooted at the favorite player. In terms of SI, the tournament digraph is the ``host'' graph and the ``pattern'' graph being sought is an {\sba} rooted at the favorite player.
In {\demandTF}, the pattern graph is an {\sba} that contains all the {\it demand arcs} (those arcs that correspond to demand matches). {To the best of our knowledge this is a novel ``edge-extension'' variant of SI where the goal is to build the pattern graph using a set of given edges as a starting point. The setting of \cite{ManurangsiSuksompong} can also be interpreted as a constrained version of SI where arcs representing matches between higher ranked players can only occur in parts of the pattern graph that represent later rounds. Their motivation was to prevent the best players from meeting too early.
	
\begin{table}[t]
    \centering
    \begin{tabular}{|l|cc|}
        \hline 
        Algorithms &$n^{O(k)}$-time   & \\
        &  $2^{\bigoh(k\log k)}n^{\bigoh(1)}$-time if $F\subseteq \Scal$ &\\
        & $3^n\cdot n^{\bigoh(1)}$-time  & \\
        \hline
        Hardness  &no $2^{o(n)}$-time algorithm (under ETH) & \\
        &~no $n^{d^{\bigoh(1)}}$-time algorithm (if NP $\not\subseteq$ QP)& \\
        \hline 
    \end{tabular}
    
    \caption{A summary of our results for {\sc Demand-TF}. Here, $n$ is the number of players, $k$ is the size of some minimum feedback arc set $F$ of the input tournament, $\Scal$ is the set of demand matches and $d$ is the number of demand matches.   
    }
\end{table}

We next describe our algorithmic contributions (see Table 1). 
On the one hand, we show that {\demandTF} is \NPH and conditionally rule out any algorithm that runs in time $n^{d^{\bigoh(1)}}$ where $d$ is the number of demand matches. 
This motivates the search for tractable restrictions of the problem and brings us to the central results  of the paper. Here, we make the following contributions: 

\smallskip
\noindent
{\bf Algorithm 1:} {\demandTF} is P-time solvable when the tournament digraph has a linear ordering on the ability of players with a {\em constant} number of exceptions where a player with lower ability beats a player with higher ability. In other words, when the {\em feedback arc set} number of the tournament digraph is constant. This is a natural condition in competitions where there is a clear-cut ranking of the players according to their skills with only a few pairs of players for which the weaker player can beat the stronger player. Motivated by empirical work in ~\cite{RussellB11}, \cite{AzizGMMSW14} initiated the design of algorithms for TF when the instances have constant feedback arc set number and gave the first P-time algorithm. This restriction was then extensively explored in a series of papers~\cite{RamanujanS17,GuptaR0Z18a,Gupta0SZ19} leading to novel fixed-parameter algorithms. 
In parameterized complexity parlance, we give an XP algorithm for {\demandTF} parameterized by the feedback arc set number $k$ (i.e., running time $n^{\bigoh(k)}$). This brings up the natural question of whether the problem is fixed-parameter tractable (FPT) (i.e., solvable in time $f(k)n^{\bigoh(1)}$ for some function $f$). Although we do not settle this question in this paper, we identify an additional structural constraint in our next result that leads to an FPT algorithm.

\smallskip
\noindent{\bf Algorithm 2:} If, in the given instance $(T,\Scal)$ of {\demandTF}, every upset match is also a  demand match, then we get fixed-parameter tractability parameterized by the feedback arc set number. The natural motivation for this scenario is a tournament designer being incentivized (e.g., by betting) to ensure that the upsets take place. 

\smallskip
\noindent{\bf Algorithm 3:} We extend our methodology in the preceding algorithms to handle further constraints (in the same running time). In particular, when the designer wants each demand match to take place in a {\em specific round} of the SE tournament, we can still find such a bracket in time $n^{\bigoh(k)}$. Such constraints allow the designer to ensure that some matches do not occur too early (in the spirit of \cite{ManurangsiSuksompong}) or too late in the tournament. 

Finally, moving to the exact-exponential-time regime we show that assuming the Exponential Time Hypothesis~\cite{ImpagliazzoP01}, one cannot get a subexponential-time algorithm for the problem (i.e., a $2^{o(n)}$ running time where $n$ is the number of players) and complement this lower bound with an algorithm with running time $2^{O(n)}$ -- an asymptotically tight bound. 

\paragraph{Organization of the paper.}
We begin by presenting basic definitions followed by our hardness results. Then, our largest section is dedicated to presenting our main algorithm ({\bf Algorithm 1}). The remainder of the algorithmic contributions ({\bf Algorithm  2} and {\bf Algorithm  3} above) are presented in the following sections. Finally, we conclude with directions for future research.

\section{Preliminaries}\label{sec:prelims}

\paragraph{Binomial arborescences.} An arborescence is a rooted directed tree such that all arcs are directed away from the root. 

\begin{definition}[\citeauthor{Williams10} \citeyear{Williams10}]\label{def:binomialArborescences}
	{\rm The set of {\em binomial arborescences} over a tournament digraph $T$ is recursively defined as follows. {\em (i)} Each $a \in V(T)$ is a binomial arborescence rooted at $a$. {\em (ii)} If, for some $i>0$, $H_a$ and $H_b$ are $2^{i-1}$-node binomial arborescences rooted at $a$ and $b$, respectively, then  adding an arc from $a$ to $b$ gives a $2^{i}$-node binomial arborescence {\rm ({\ba})} rooted at $a$.
	If a binomial arborescence $H$ is such that $V(H)=V(T)$, then $H$ is a {\em spanning binomial arborescence} {\rm ({\sba})} of $T$.}
	
\end{definition}

The relevance of binomial arborescences comes from the following variant of a result of \cite{Williams10}.

\begin{proposition}\label{prop:equivalence}Let $T$ be a tournament digraph and let ${\mathcal S}\subseteq A(T)$. Then, there is a seeding of $V(T)$ such that the resulting SE tournament has every match in $\mathcal{S}$ if and only if $T$ has an {\sba} $H$ such that $A(H) \supseteq \mathcal{S}$.
\end{proposition}

We use the terms vertex and player interchangeably.
A {\em rooted forest} is the disjoint union of a set of arborescences.
Let $H$ be a rooted forest. For a vertex $u\in V(H)$, denote by $\child_H(u)$ the set of children of $u$ in $H$, by $\descendant_H(u)$ the set of descendants of $u$ in $H$ (including $u$).
The {\em strict} descendants of $u$ comprise descendants of $u$ that are neither $u$ nor children of $u$. 
We denote by $\sibling_H(u)$, the set of siblings of $u$ in $H$.
We define the \emph{height} of $v$ in $H$, $\alpha_H(v) := \log|\descendant_H(v)|$.
Note that if $H$ is an {\sba} of $T$, then $\forall v\in V(T),~\exists i\in [\log n]\cup \{0\}$ such that $|\descendant_H(v)|=2^i$ and hence $\alpha_H(v)$ is an integer.
The interpretation in the corresponding SE tournament is that $v$ is the winner of a subtournament played by the players who are descendants of $v$ in $H$ and $\alpha_H(v)$ is the number of matches that $v$ wins.
If $\alpha_H(v)>\alpha_H(u)$, then we say that $v$ is {\em higher} than $u$ and $u$ is {\em lower} than $v$. 
We will also refer to the \emph{height} of a \ba meaning the height of its root.

We will use the following characterization of {\ba}s:
\begin{proposition} \label{prop:altba}
For $n>0$, $H_n$ is a {\ba} of height $n$ rooted at $v_n$ if and only if $H_n = \bigcup_{i=0}^{n-1}{H_i}$ where $H_i$ is a {\ba} of height $i$ rooted at $v_i$ together with an edge from $v_n$ to each $v_i$.
\end{proposition}

In an instance $(T,\mathcal{S})$ of {\sc Demand-TF}, we call the arcs in $\mathcal{S}$ {\em demand arcs} or {\em demand matches} and their endpoints {\em demand vertices}. 
For every demand arc $(p,q)$, we say that $p$ is a {\em demand in-neighbor} or {\em demand parent} of $q$ and $q$ is a {\em demand out-neighbor} or {\em demand child} of $p$. We will use demand parent and demand child in the context of rooted trees and demand in/out-neighbor otherwise. For a vertex $p$, the number of its demand in-neighbors is called its {\em demand in-degree}. The {\em demand out-degree} is defined symmetrically. 
$\head(\cS)$ denotes the vertices with a demand in-neighbor, that is, those vertices that lose some demand match. 

\section{Hardness Results for {\sc Demand-TF}}

We first show that {\sc Demand-TF} is NP-complete and then infer further facts regarding the complexity of this problem using our proof in combination with hardness results on {\TFP} proved by \cite{AzizGMMSW14}.  

\begin{theorem}\label{thm:demandTFPIsNPComplete}
	{\sc Demand-TF} is {\npc}. 
\end{theorem}

\begin{proof}
To demonstrate the NP-hardness, we will give a polynomial-time reduction from the {\sc Tournament Fixing} problem (TF). Let $(T, v^*)$ be an instance of TF. Let $n = |V(T)|$. We first construct an $n$-vertex acyclic tournament, $D_1$ with source vertex $d_1$. We call $D_1$ a ``dummy'' tournament.
Notice that in any SE tournament played by the vertices in $D_1$ regardless of the seeding, the vertex $d_1$ will be the winner.
Now, define $T'$ to be the tournament obtained by taking the disjoint union of  $T$ and $D_1$ and then doing the following:
\begin{enumerate}
	\item Add the arc $(d_1, v^*)$, ensuring that $v^*$ loses to $d_1$.
	\item For every other vertex $v\in V(T)$, add the arc $(v, d_1)$, ensuring that $d_1$ loses to every vertex in $T$ except $v^*$. 
	\item Add arcs ensuring that every vertex in $T$ beats every vertex in $V(D_1)$ except for $d_1$.
\end{enumerate}
We initially set $\mathcal{S}:= \{(d_1,v^*)\}$, find an arbitrary  
{\sba} $H_1$ on $D_1$ and for each arc of $H_1$ incident on $d_1$, we add this arc to $\mathcal{S}$. This completes the construction of the {\sc Demand-TF} instance $(T',\mathcal{S})$.

Clearly, the reduction can be done in polynomial time, so it remains to argue the correctness. 

Suppose that $(T,v^*)$ is a  yes-instance of TF. Then there exists a permutation $\pi$ over $V(T)$ such that $v^*$ wins the SE tournament where the first-round matches are given by $\pi$ and the pair-wise results by $T$. Suppose that the permutation of $V(D_1)$ that leads to the {\sba} $H_1$ is $\pi_1$. Then, notice that the permutation $\pi'$ obtained by simply taking the union of $\pi$ and $\pi_1$ is a permutation of $V(T')$. We claim that $\pi'$ certifies that $(T',\mathcal{S})$ is a yes-instance of {\sc Demand-TF}.  That is, in the SE tournament where the first-round matches are given by $\pi'$ and the pair-wise results by $T'$, every demand match is played. 
By construction, $\pi'$ results in $v^*$ losing to $d_1$ in the final. Moreover, the second-half of the bracket given by $\pi'$ that comprises only of vertices from $V(D_1)$ is identical to $\pi_1$, guaranteeing that all matches corresponding to arcs in the {\sba} $H_1$ are indeed played. This completes the forward direction. 

Conversely suppose  that there exists a permutation $\pi'$ on $V(T')$ which results in every demand match being played. Notice that the definition of $\mathcal{S}$ implies that $d_1$ wins the tournament and plays exactly $\log n+1$ matches, out of which one is against $v^*$ and the remaining $\log n$ matches must be against vertices of $D_1$. We next argue that $d_1$ cannot play against $v^*$ before the final. If this were not the case, then one half of the bracket given by $\pi'$ contains both $d_1$ and $v^*$. Then, the other bracket must contain at least one vertex from $V(T)$. Since all vertices in $V(T)\setminus \{v^*\}$ are stronger than every vertex in $V(D_1)$, it follows that the opponent of $d_1$ in the final would have to be a vertex of $T$ other than $v^*$, a contradiction to $d_1$ winning the whole tournament. By the same reasoning, the half of the bracket given by $\pi'$ that contains $d_1$ must in fact only comprise vertices of $D_1$. This implies that $v^*$ wins its subtournament that comprises exactly the vertices of $T$, implying that $(T,v^*)$ is a yes-instance of TF. 
\end{proof}

\begin{theorem}\label{thm:consequencesOfNPHardnessReduction}
    From Theorem~\ref{thm:demandTFPIsNPComplete} combined with known hardness of TF \cite{AzizGMMSW14}, we obtain the following.
    \begin{enumerate}
        \item 	{\sc Demand-TF} is {\npc} even if there is a vertex on which every demand arc is incident. 
        \item Unless ${\rm NP}\subseteq {\rm Quasi}$-P-time, there is no algorithm for {\sc Demand-TF} that runs in time $n^{d^{\bigoh(1)}}$, where $d$ is the number of demand arcs. In particular, this rules out a  $2^{d^{\bigoh(1)}}n^{\bigoh(1)}$-time fixed-parameter algorithm. 
        \item Assuming the Exponential Time Hypothesis (ETH), there is no $2^{o(n)}$-time algorithm for {\sc Demand-TF}. 
    \end{enumerate}
\end{theorem}

\begin{proof}
    
    The first statement follows from the construction in Theorem~\ref{thm:demandTFPIsNPComplete}. 	
    Moreover, notice that in the same construction, the number of demand arcs is $\log n+1$. Hence, an algorithm for {\sc Demand-TF} with running time $n^{d^{\bigoh(1)}}$ would imply a quasi-polynomial-time algorithm for TF, which is NP-complete. 
    %
    Finally, the proof of NP-hardness of TF given in \cite{AzizGMMSW14} reduces an instance of 3-SAT-2L (3-SAT where every literal appears at most twice) with $n$ variables to an instance of TF with $\bigoh(n)$ vertices. Since our reduction from TF
    only doubles the number of vertices, a $2^{o(n)}$-time algorithm for {\sc Demand-TF} would imply the same running time for 3-SAT-2L, which violates ETH.
\end{proof}

Recall that the naive algorithm for {\sc Demand-TF} has running time $2^{\bigoh(n\log n)}$ as a result of brute-forcing over all possible brackets. We next improve this to a single-exponential running time with a dynamic programming algorithm, {\em asymptotically} matching the $2^{o(n)}$-time lower bound from Theorem~\ref{thm:consequencesOfNPHardnessReduction}. 
	
\begin{theorem}\label{thm:exactExponentialTimeAlgorithm}
    {\sc Demand-TF} can be solved in time $3^nn^{\bigoh(1)}$. 
\end{theorem}

\begin{proof}
	Let $(T,\cS)$ be the given instance of {\sc Demand-TF}. Let us define a boolean function $\Delta:2^{V(T)}\times V(T)$ as follows. For every $S\subseteq V(T)$ and $x\in V(T)$, $\Delta(S,x)=1$ if and only if the following conditions are satisfied.
	\begin{itemize}
		\item $|S|$ is a power of 2 and $x\in S$. 
		\item There is a permutation $\pi_S$ over $S$ such that the SE tournament played by $S$ according to the seeding $\pi_S$ is won by $x$ and every demand match in $\cS\cap A[S]$ is played in this tournament.  
	\end{itemize} 
	
	Notice that $(T,\cS)$ is a yes-instance if and only if there is some $x\in V(T)$ such that $\Delta(V(T),x)=1$. Hence, it suffices to give an algorithm that computes the function $\Delta$ in the stated running time.  We have the following claim at the crux of our algorithm.  
	
	\begin{cclaim}
		Consider a vertex set $S$ of size $2^{p}$ for some $p>0$,  and let $x\in S$. Then,  $\Delta(S,x)=1$ if and only if  there is an  equi-partition $S_1\uplus S_2=S$ and a vertex $y$ such that $\Delta(S_1,x)=1$,  $\Delta(S_2,y)=1$, $(x,y)\in A(T)$ and except for the arc $(x,y)$ there is no demand arc in $\mathcal{S}$ with one endpoint in $S_1$ and the other endpoint in $S_2$. 
	\end{cclaim}
	
	\begin{proof}
			In the forward direction, suppose that $\Delta(S,x)=1$ and consider an SE tournament played by $S$ that is won by $x$ and in which every demand match in $\cS\cap A[S]$ is played. Let $y$ be the opponent of $x$ in the final of this SE tournament. Let $S_2$ be the descendants of $y$ in the corresponding {\sba} and let $S_1=S\setminus S_2$. Then, notice that except for potentially the arc $(x,y)$, there cannot be a demand arc with one endpoint in $S_1$ and the other in $S_2$. Moreover, $(x,y)\in A(T)$ and every demand arc contained in $\cS\cap (A[S_1]\cup A[S_2])$ must be played, i.e., $\Delta(S_1,x)=\Delta(S_2,y)=1$. This completes the proof of the forward direction. The argument for the converse is symmetrical.  
	\end{proof}
	
	Given the above claim, we compute $\Delta$ using a dynamic programming algorithm as follows. In the base case, when $|S|=1$, this is trivial. Now, suppose that we have computed $\Delta(S,x)$ for every $|S|\leq 2^i$. Consider a vertex set $S'$ of size $2^{i+1}$,  and let $x\in S'$. Using the above claim, it is sufficient to go through every possible equi-partition of $S'$ into sets $S_1'$ and $S_2'$ and check whether there exists a vertex $y$ such that  $\Delta(S_1',x)=1$,  $\Delta(S_2',y)=1$, $(x,y)\in A(T)$ and except for the arc $(x,y)$ there is no demand arc in $\mathcal{S}$ with one endpoint in $S'_1$ and the other endpoint in $S'_2$. This can clearly be done in time $\bigoh^*(2^{|S'|})$. Consequently, the running time of our overall algorithm is bounded by $\Sigma_{i=0}^n{n \choose i}2^in^c$ for some constant $c$, implying a running time of $\bigoh^*(3^n)$ as claimed. 
	This completes the proof of the theorem. 
\end{proof}

\subsection{Going beyond Theorem \ref{thm:exactExponentialTimeAlgorithm}}
We first remark that similar to \cite{KimW15}, the subset convolution technique can be used to speed up the algorithm of Theorem \ref{thm:exactExponentialTimeAlgorithm} to a $2^nn^{\bigoh(1)}$-time algorithm. 
We also point out that the approach used in the  the algorithm of Theorem~\ref{thm:exactExponentialTimeAlgorithm} can be easily extended to a more general version of the problem where we assign to the matches, arbitrary (non-negative) integer weights that are polynomially bounded, and we want to compute a seeding that {\em maximizes} the total weight of the satisfied demands. To achieve this, we would simply have to enhance the function $\Delta$ so that for every $S\subseteq V(T), x\in V(T)$ and polynomially bounded $w\in {\mathbb N}$,  $\Delta(S,x,w)=1$ if and only if  $|S|$ is a power of 2, $x\in S$ and there is a permutation $\pi_S$ over $S$ such that the SE tournament played by $S$ according to the seeding $\pi_S$ is won by $x$ and total weight of the demand matches in $A[S]$ that are played in this SE tournament is $w$. Then, the analogue of the claim in the proof of Theorem~\ref{thm:exactExponentialTimeAlgorithm} would be: $\Delta(S,x,w)=1$ if and only if there is an  equi-partition $S_1\uplus S_2=S$ and integers $w_1,w_2$, and a vertex $y$ such that $\Delta(S_1,x,w_1)=1$,  $\Delta(S_2,y,w_2)=1$, $(x,y)\in A(T)$ and the sum of $w_1,w_2$ and the weight assigned to the arc $(x,y)$ adds up to $w$. 

\section{{\sc Demand-TF} on Graphs of Bounded Feedback Arc Set Number}
We now turn our attention to tournament digraphs with a constant-size (denoted by $k$) feedback arc set.  In this section we will assume that an instance of {\demandTF} is a triple $(T, \mathcal{S},F)$, where $F$ is a minimum feedback arc set of $T$. The assumption that $F$ is given, is without loss of generality since a minimum feedback arc set of size at most $k$ can be computed in time $3^kn^{\bigoh(1)}$-time~\cite{CyganFKLMPPS15}. We refer to the endpoints of $F$ as {\em feedback vertices}. Additionally we will assume that $\sigma=v_1,v_2, \dots, v_n$ is a linear ordering of the vertices of $T$ such that the arcs $(v_i,v_j)$ with $i>j$ are precisely the ``upset'' matches, i.e., the arcs of $F$. We say that $v_i$ is \emph{stronger} than $v_j$ for all $j > i$. 
We say that a subgraph of the tournament digraph is {\em valid} if every demand arc between two vertices in the subgraph is also present in the subgraph.  Thus, Proposition~\ref{prop:equivalence} implies that a solution to {\sc Demand-TF} corresponds to a valid {\sba} of the given tournament digraph. 

Note that any player can lose at most one match in an SE tournament. 
Hence, if $(T,\mathcal{S},F)$ is an instance of {\sc Demand-TF} in which some vertex has demand in-degree greater than 1, then it is a no-instance. So, we may assume without loss of generality that in any non-trivial instance of {\sc Demand-TF}, every vertex has at most one demand in-neighbor.  In the rest of this section we will also assume that all but at most one of the feedback vertices has a demand in-neighbour: this is true if, in the final \sba, {\em every} feedback vertex except the root has a demand parent. We will ensure this in our algorithm by guessing the parent of every feedback vertex (i.e., the player that beats it) and adding the resulting arc to the set of demand arcs (the overhead is at most $n^{2k}$). 

We also guess the heights of each feedback vertex (overhead $(\log n)^{2k}$). Using this as a starting point, we obtain an estimate for the heights of every vertex via the following definition, where the reader may think of the function $g$ as our guesses for the heights of the feedback vertices. 

\begin{definition}[Function $\alpha^*$ and compactness property]\label{def:alphastarwithg}
	{\rm Fix a function $g : V(F) \rightarrow [\log n]$.
	For each $v \in V(T)$ let $\alpha^*_g(v) = g(v)$ if $v \in V(F)$.
	Otherwise let $\alpha^*_g(v)$ be the minimum non-negative integer satisfying:
	\begin{enumerate}
		\item  For each demand arc $(v, w)$, $\alpha^*_g(v) > \alpha^*_g(w)$. 
		\item  For each $u,w$ such that $(u, v), (u, w) \in \mathcal{S}$, $\alpha^*_g(v) \neq \alpha^*_g(w)$ if either (i) $w$ is weaker than $v$ or (ii) $w \in V(F)$. 
	\end{enumerate}
	We say that a binomial arborsescence $H$ is \emph{compact with respect to $g$} if $\alpha_H(v) = \alpha^*_g(v)$, for every $v \in \head(\mathcal{S}) \cap V(H)$.\\
	Additionally, we say that a {\ba} $H$ is \emph{weakly compact} if it is compact with respect to the function $\alpha_H$ restricted to $V(F)$.}
\end{definition}

Intuitively, $\alpha^*_g(v)$ can be described as follows. Fix a hypothetical solution, that is a valid {\sba} $H$, and suppose that in $H$, we know the heights of each demand child and each weaker demand sibling of $v$. Moreover, suppose that we know the heights of the vertices in $V(F)$, which is expressed by the function $g$. Based on this information, since $H$ contains every demand arc, one can narrow down the set of all possible heights that $v$ can have in $H$, e.g., by using the fact that $v$ is higher than every child, $v$ cannot have the same height as a sibling, and so on. The value of $\alpha^*_g(v)$ is the smallest candidate value of the height of $v$ we are left with. 
In other words, $\alpha_g^*(v)$ gives a lower bound on the height of $v$ in {\em any} solution. This is formally stated below. 

\begin{observation} \label{obs:atleastalphastar}
	If a valid \ba $H$ is compact with respect to $g$, then, for all $v \in V(H)$, $\alpha_H(v) \geq \alpha_g^*(v)$.
\end{observation}

\begin{proof}
    Clearly the claim holds for $v \in \head(\mathcal{S})$ by definition. If $v$ is not a demand vertex, that is there are no demand edges incident on $v$, then $\alpha_g^*(v) = 0$ so the claim holds vacuously. Otherwise we have $\alpha_H(v) = 1 + \max_{w \in \child_H(v)}{\alpha_H(w)}$ and $\alpha_H(w) = \alpha_g^*(w)$ whenever $(v, w) \in \mathcal{S}$ since then $w \in \head(\mathcal{S})$. Additionally condition 2 never applies so $\alpha_H(v) \geq \alpha_g^*(v)$.
\end{proof}

Note that for any vertex $v$, $\alpha^*_g(v)$ is completely determined by only 
$\alpha^*_g(w)$ where $w$ is a child of $v$ or a sibling that is weaker or a feedback vertex. If $w$ is a feedback vertex then $\alpha^*_g(w)$ is determined by $g$, otherwise in both of the other cases, $w$ is weaker than $v$. So $\alpha^*_g$ can be easily calculated in polynomial time by simply applying the definition to vertices in strength order beginning with the weakest. From now on we will assume that we know $\alpha^*_g$.

\textbf{The central insight behind our algorithm} that leads to the notion of compactness is that in yes-instances, there is always a solution where the height of {\em every vertex} that loses a demand match is precisely this smallest candidate value (this is formalized in Lemma \ref{lem:alwaysweaklycompact}). This motivated our definition of weak compactness in Definition~\ref{def:alphastarwithg}. Given this fact, our algorithmic strategy is to ``pack'' the rest of the vertices into the solution {\sba} using an intricate subroutine that is guided by this insight. 
Roughly speaking, our algorithm will use a greedy approach to complete the packing, where at any step, a set of partially constructed subgraphs are available and the goal is to make a new partial solution that contains the latest vertex that is processed. However, the challenge our approach has to face is that in the intermediate steps of our algorithm we would be  dealing with partially constructed subgraphs (i.e., forests) where each component is not necessarily an \sba, yet we cannot simply break them apart since they encode important height information that we wish to enforce in the complete solution. Thus, the step-by-step challenge, solved by subroutine \textsc{Pack} which we describe later, is to carefully ``glue'' some of these structures together to form a supergraph in each step such that in the final step we have a \ba. The trickiest aspect is to do this in such a way that if at any point we cannot find appropriate pieces to glue together, we are able to correctly reject.

\paragraph{Guaranteed compactness.} 
We next formalize our central insight and prove that every yes-instance has a valid {\em weakly compact} {\sba}. Towards this, we  argue that we can modify any valid {\sba} to achieve this property using the following ``exchange'' lemma.

\begin{lemma} \label{lem:almostAcyclicExchange}
	Suppose $H$ is a valid \sba, $(u, v) \in \mathcal{S}$, $w \in \sibling_H(v)$ such that  $v, w \notin V(F)$, $w$ is lower than $v$ and let $B$ be the set of children of $v$ that are at least as high as $w$. 
	Moreover, suppose that $v$ has no demand out-neighbors in  $B$ and either (i)  $w$ is stronger than $v$ or (ii) $(u, w) \notin \mathcal{S}$. 
	Then, $H$ can be transformed to a valid {\sba}, $H'$ where the following hold:
	\begin{enumerate} \item  $\alpha_{H'}(v) = \alpha_{H}(w)$. That is, after the transformation, $v$ now has the ``old'' height of $w$.  
		\item  For every $x\notin B\cup \{v,w\}$, $\alpha_{H}(x)=\alpha_{H'}(x)$. That is, except for a few vertices adjacent to $v$ in $H$ the heights of all other vertices remain the same after the transformation.  
	\end{enumerate}
\end{lemma}

\begin{proof}
    First of all, we observe that 
	no vertex of $B$ can be in $V(F)$. Indeed, since the vertices in $B$ are children of $v$ and $v$ has no demand out-neighbors in $B$ by the premise of the lemma, it follows that the vertices in $B$ have no demand in-neighbors. On the other hand, recall that we have assumed that every non-root feedback vertex has a demand in-neighbor. Hence, $B$ is disjoint from $V(F)$.

	The vertices of $B$ are weaker than $v$ which in turn must be weaker than $u$ since $v\notin V(F)$. Similarly, $w$ is weaker than $u$ because of the premise that $w\notin V(F)$. Moreover, notice that no vertex in $B$ can beat $u$. Otherwise, we would have a 3-cycle containing $u$,$v$ and a vertex of $B$, implying that two of these vertices must be in $V(F)$, a contradiction to the preceding arguments. More generally, $u$ beats $B\cup \{v,w\}$, it must be the case that if the subgraph induced by $B \cup \{u, v, w\}$ contains a cycle, then $V(F)$ contains at least two vertices of $B \cup \{u, v, w\}$. But this is a contradiction to the premise that $\{v,w\}\notin V(F)$ and our earlier conclusion that $B$ is disjoint from $V(F)$. Hence, the subgraph induced by $B\cup \{u,v,w\}$ is acyclic.
	
	Let $\ell=\alpha_H(v)-\alpha_H(w)-1$ and let $r_0,\dots, r_\ell$ be the vertices in $B$, where, for each $i\in \{0,
	\dots, \ell\}$, $\alpha_H(r_i)=\alpha_H(w)+i$. Moreover, for each $i\in \{0,\dots, \ell\}$, let $T_i$ denote the {\ba} of height $\alpha_H(w)+i$, rooted at $r_i$. 
	
	Recall that we have assumed that $v$ has no demand out-neighbors in $B$. So,  we can remove the arc $(v, x)$ for each $x \in B$ without ``losing'' any demand matches. Let $H'$ be the rooted forest obtained from $H$ by removing these arcs. We will now modify $H'$ to obtain an {\sba} where the height of $v$ is the same as the height of $w$ in $H$. 
	Notice that currently, $\alpha_{H'}(v) = \alpha_H(w)$ and the vertices of $B$ are now the roots of {\ba}s of heights $\{\alpha_H(w)+i\mid i\in \{0,\dots, \ell\}\}$.  
	
	We now consider the following two cases. 
	\begin{description}[style=unboxed,leftmargin=0cm]
        \item [Case 1:] $(u, w) \in \mathcal{S}$. Then, we are in the case where $w$ is stronger than $v$ and so, $w$ beats $r_0,\dots, r_{\ell}$. Hence, we can add to $H'$, the arc $(w, x)$ for each $x \in B$, thereby rooting these {\ba} below $w$ and converting $H'$ to an {\sba}. This ensures that $\alpha_{H'}(v)=\alpha_H(w)$. Effectively, we have swapped the heights of $v$ and $w$.
        \item [Case 2:] $(u, w) \notin \mathcal{S}$. In this case, remove $(u,w)$ from $H'$ and call the resulting forest $H''$. Notice that $\alpha_{H''}(v)=\alpha_H(w)$. Therefore, it remains to construct a {\ba} $Q$ of height $\alpha_H(v)$ using the trees $T_0,\dots, T_\ell$ along with the subtree of $H$ rooted at $w$ (call it $T_w$), such that $Q$ can be rooted below $u$.  Towards this, notice that $u$ beats every vertex in $B\cup \{w\}$. Hence, our task is simply to ``pack'' the trees $T_w,T_0,\dots, T_\ell$ into a {\ba} of height $\alpha_H(v)$ and as long as the root of this {\ba} is contained in $B\cup \{w\}$, we can just root this {\ba} below $u$. 
		
		\begin{figure}[t]
			\begin{center}
				\includegraphics[scale=0.9]{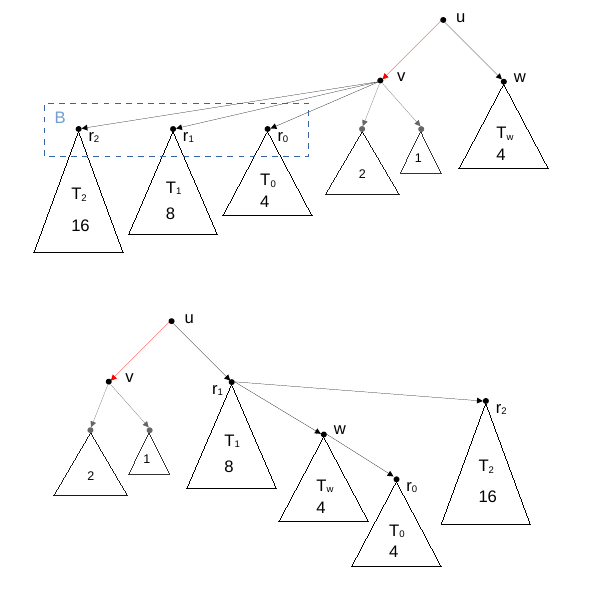}
				\caption{An illustration of the exchange operation in Case 2 in the proof of Lemma~\ref{lem:almostAcyclicExchange}. The first image is before the transformation and the second image is after. Notice that $(u,w)$ is not a demand arc in this case. Moreover, in this figure, we are assuming that $r_1$ beats $w$ and $r_2$ while $w$ beats $r_0$. The exchange argument in Case 1 is straightforward. We simply make $r_0,r_1,r_2$ children of $w$. }
				\label{fig:exchangeLemma1}
			\end{center}
		\end{figure}
		
		Set $p_0=w$ and for each $i\in [\ell]$, define the vertex $p_i$ as the stronger vertex in the pair $\{p_{i-1},r_{i-1}\}$ and $s_i$ as the weaker vertex. 
		Similarly, let $J_0$ denote the subtree of $H$ rooted at $w$ and for every $i\in [\ell+1]$, define $J_i$ to be the {\ba} obtained by taking $J_{i-1}$ and $T_{i-1}$ and making $p_i$ the root, i.e., adding the arc $(p_i,s_i)$. Notice that $J_1,\dots,J_{\ell+1}$ exist since, for each $i\in \{0,\dots, \ell\}$, the trees $J_i$ and $T_i$ are {\ba} of height $\alpha_{H}(w)+i$. Moreover, $J_{\ell+1}$ is a {\ba} of height $\alpha_{H}(w)+\ell+1=\alpha_H(v)$ and is rooted at $p_{\ell+1}\in B\cup \{w\}$, which is weaker than $u$. Hence, we can simply delete from $H''$ the trees $J_0,T_0,\dots, T_\ell$, add the {\ba} $J_{\ell+1}$ and make its root $p_{\ell+1}$, a child of $u$. 
		See Figure~\ref{fig:exchangeLemma1} for an illustration of this exchange operation. 
	\end{description}
	
	In both cases, it is straightforward to check that the second statement of the lemma holds.
\end{proof}

\begin{lemma} \label{lem:alwaysweaklycompact}
	If there exists a valid {\sba} $H$, then there exists a valid, weakly compact {\sba}.
\end{lemma}
\sv{
	\begin{proof}[Proof Sketch]
		In any non-weakly-compact \sba there is a weakest vertex that contradicts the definition of weak compactness. We choose a valid \sba that has the strongest such ``certificate of non-compactness'' and aim to find another \sba with a stronger one which would give the required contradiction. By properties of {\ba}s this certificate vertex has a sibling that is the correct height so we apply Lemma \ref{lem:almostAcyclicExchange} to swap these heights. This gives us the required \sba.
	\end{proof}
}
\begin{proof}
    Note that if a valid {\sba} $H$ is not weakly compact, then  there is a weakest vertex $v\in \head(\Scal)$ such that  $\alpha_H(v) \neq \alpha_g^*(v)$, where $g$ is the restriction of $\alpha_H$ to $V(F)$. Clearly, $v$ cannot be a feedback vertex since by definition of $g$, $\alpha_H(v)=\alpha_g^*(v)$. We call $v$ the {\em certificate of non-compactness} of $H$. In the rest of the proof of this lemma, we assume that
	$H$ is chosen in such a way that its certificate of non-compactness $v$ is the strongest possible among those of all valid, non-weakly-compact {\sba}. 
	Then, $\alpha_H(x) = \alpha_g^*(x)$ for every $x$ that is weaker than $v$ and loses a demand match.
 
	We aim to find an $H'$ with a strictly stronger certificate of non-weak-compactness, that is $\alpha_{H'}(x) = \alpha^*_{g'}(x)$ for all $x \in \head(\mathcal{S})$ that is either $v$ or weaker than $v$, where $g'$ is the restriction of $\alpha_{H'}$ to $V(F)$.  This would give us the required contradiction.
 
	First of all, $v$ must be higher than all of its children in $H$ and $v$ cannot have the same height as any of its siblings in $H$. So, $\alpha_H(v) \neq  \alpha_g^*(v) $ implies that $\alpha_H(v) > \alpha_g^*(v)$. Moreover,  since $H$ is a {\ba}, there must exist $w \in \sibling_H(v)$ such that $\alpha_H(w) = \alpha_g^*(v)$.  Let $u$ be the parent of $v$ and $w$ in $H$. Moreover, $w\notin V(F)$ since otherwise, $\alpha_g^*(w)$ and $\alpha_g^*(v)$ would coincide, which is not possible.
	
	Our next goal is to prove the following claim, which will enable us to use our exchange arguments from Lemma~\ref{lem:almostAcyclicExchange}.
	
	\begin{cclaim}
		$H,\mathcal{S},u,v,w$ satisfy the premise of Lemma~\ref{lem:almostAcyclicExchange}. 
	\end{cclaim}  

	\begin{proof}
    	We have already argued that $\alpha_H(v) > \alpha_g^*(v)$ and $\alpha_H(w) = \alpha_g^*(v)$, implying that  $\alpha_H(w)<\alpha_H(v)$, i.e., $w$ is lower than $v$ in $H$.
    	
        Now, suppose for a contradiction that  $v$ has, among its children that are at least as heavy as $w$ (i.e., the set $B$), a demand-out-neighbor $v'$.  Since $v\notin V(F)$, it follows that the arc $(v,v')$ is not a feedback arc, implying that $v'$ is weaker than $v$. By our selection of $v$ as a certificate of non-compactness, it follows that $\alpha_H(v')=\alpha_g^*(v')$. This implies (by invoking Condition 1 in Definition~\ref{def:alphastarwithg}) that:  $$\alpha_g^*(v)>\alpha_g^*(v')=\alpha_H(v')\geq \alpha_H(w)=\alpha_g^*(v).$$
        
        This is a contradiction, hence we conclude that $v$ has no demand out-neighbors in the set $B$. 
        
        Finally, it remains to argue that either $w$ is stronger than $v$ or $(u,w)$ is not a demand match. We will argue that the former holds. Indeed, if $w$ is weaker than $v$, then this would contradict Condition 2 in Definition~\ref{def:alphastarwithg} because $\alpha_g^*(v)$ could not have  $\alpha_H(w)$ as a possible candidate value. Hence, Lemma~\ref{lem:almostAcyclicExchange} is applicable. 
    \end{proof}

	Now, let $H'$ be the valid {\sba} obtained by invoking Lemma~\ref{lem:almostAcyclicExchange} on $H,\mathcal{S}, u,v,w$. Then, we have that  $\alpha_{H'}(v) = \alpha_H(w) = \alpha^*(v)$. We have the following claim:
	
	\begin{cclaim} \label{claim:almostAcyclicWeakerUnchanged}
		For every $x$ that is weaker than $v$ and loses a demand match, $\alpha_{H'}(x)=\alpha^*(x)$. 
	\end{cclaim} 
	
	\begin{proof}
		The second statement of Lemma~\ref{lem:almostAcyclicExchange} guarantees that for every $x$ that is weaker than $v$ and disjoint from $B\cup \{v,w\}$, we have that $\alpha_{H'}(x) = \alpha_{H}(x)$. Here, $B$ is the set of children of $v$ that are at least as high as $w$ in $H$.  Since we already know that for every such $x$, $\alpha_{H}(x)=\alpha^*(x)$, we conclude that $\alpha_{H'}(x)=\alpha^*(x)$. 
		
		It remains to argue that for every $x\in B\cup \{w\}$ that is weaker than $v$ and loses a demand match, $\alpha_{H'}(x)=\alpha^*(x)$.
		We have already argued that the vertices in $B$ must have a demand in-degree of 0. Hence, $x\notin B$. We next consider the possibility that $x=w$. Recall that since Lemma~\ref{lem:almostAcyclicExchange} was applicable on $H,\mathcal{S}, u,v,w$, we know that  either $w$ has a demand in-degree of 0 or $w$ is stronger than $v$. In either case, we have that $x\neq w$, completing the proof of the claim.
	\end{proof}
	
	By definition $v \notin V(F)$ and we know that $w \notin V(F)$ since every feedback vertex has a demand parent and if $w$ shared a demand parent with $v$ then condition 2 of Definition~\ref{def:alphastarwithg} would have ensured that $\alpha^*_g(v) \neq \alpha_H(w)$.
	So, we can apply Lemma~\ref{lem:almostAcyclicExchange} to get $H'$ where $\alpha_{H'}(v) = \alpha_H(w) = \alpha^*_{\alpha_H}(v)$. Condition 2 of Lemma~\ref{lem:almostAcyclicExchange} ensures that $\alpha_H$ only differs from $\alpha_{H'}$ on $B \cup \{v, w\}$, which is disjoint from $V(F)$. Hence $\alpha^*_{\alpha_H} = \alpha^*_{\alpha_{H'}}$ so $\alpha_{H'}(v)  = \alpha^*_{\alpha_H}(v)$.
	Claim~\ref{claim:almostAcyclicWeakerUnchanged} still holds in this setting so $H'$ has a stronger certificate of non-weak-compactness which is a contradiction.
\end{proof}

Before moving to the description of our packing subroutine, we need to define a natural relaxation of {\ba} to account for feedback vertices. 

\paragraph{Partial Binomial Arborescences.}
In our greedy packing strategy, we will process vertices from weakest to strongest, with the underlying assumption that when we arrive at a vertex $v$, all descendants of $v$ are weaker and have their respective sub-arborescences built in an earlier step. When $T$ is acyclic this works fine, however this does not go smoothly when we have cycles since a descendant of the current vertex, $v$, may actually be stronger than $v$.  Hence it is not yet processed by our algorithm, and consequently its sub-arborescence is not yet built. This leads to the scenario that the \ba  in the solution that is rooted at $v$ cannot be fully built either. Notwithstanding this difficulty, we note that the partial structures our algorithm deals with have enough \ba-like properties that one direction of Proposition 2 still holds. This leads us to the following definition, which relaxes the conditions of a \ba when feedback vertices are encountered.

\begin{definition}[Feedback descendants and partial binomial arborescence] \label{def:pba}
	{\rm Given a rooted forest $Q$ and a vertex $v \in V(Q)$, define the \emph{feedback descendants of $v$ in $Q$},
	$$\fdesc_Q(v) := \bigcup_{f \in \descendant_{Q}(v) \cap V(F), f \neq v}{\descendant_{Q}(f) \setminus \{f\}}$$
	
	Given a {\ba} $H$ of height $i$, on a tournament digraph $T$ with feedback arc set $F$, we call a subtree $H'$ of $H$ a \emph{partial binomial arborescence (\pba) of height $i$} if 
	$V(H)\setminus V(H') \subseteq \fdesc_H(\troot(H))$.}
\end{definition}
In Definition~\ref{def:pba}, 
$\fdesc_Q(v)$ are strict descendants of a feedback vertex that is itself a strict descendant of $v$. Equivalently these are vertices $x$ where the path from $v$ to $x$ contains a feedback vertex that is not $v$ or $x$, that is, $x$ is ``past'' a feedback vertex. A \pba is a \ba that is missing some feedback descendants of its root. Figure \ref{fig} contains an example of feedback descendants.

\begin{observation} \label{cor:altpba}
	For $n>0$, if $H_n = \bigcup_{i=0}^{n-1}{H_i}$ where $H_i$ is a {\pba} of height $i$ rooted at $v_i$ together with an edge from $v_n$ to each $v_i$ then $H_n$ is a {\pba} of height $n$ rooted at $v_n$.
\end{observation}

\paragraph{Guessed size.}
Definition \ref{def:pba} defines the height of a \pba implicitly. Our algorithm constructs an \sba by gluing {\pba}s of specific heights together so it needs their heights, or equivalently an appropriate notion of their size, to see if there are any good candidate {\pba}s to glue together. The following definition allows us to calculate the size of a \pba.

\begin{definition}
{\rm 	Given a rooted forest $Q$, and  $g : V(F) \rightarrow [\log(n)]$ define the \emph{guessed size of $v$},
	$$\beta_{Q,g}(v) = \begin{cases}
		2^{g(v)}  & \text{ if } v \in V(F)\\
		1 + \sum_{w \in \child_Q(v)}{\beta_{Q,g}(w)} & \text{ otherwise}
	\end{cases}$$
	We will drop the reference to $g$ when it is clear from the context.}
\end{definition}

Note that if $H$ is an {\sba} and $g$ is $\alpha_H$ restricted to $V(F)$, then $\beta_{H, g}(v) = |\descendant_H(v)| = 2^{\alpha_H(v)}$ for all $v$.
Furthermore if $H' \subset H$ is a {\pba} of height $i$ rooted at $v$ then $\beta_{H', g}(v) = 2^i$.
Effectively, deleting feedback descendants of $v$ does not change the value of $\beta(v)$.
Note that the converse is not true so we will need to prove that a given subgraph is a \pba and $\beta$ will then check its height.

Since $Q$ is a rooted forest, each vertex is the child of at most one vertex. So applying the definition recursively will only require calculating $\beta_Q(w)$ once for each vertex $w \in V(Q)$. Therefore $\beta_Q(v)$ can be calculated in polynomial time for any $Q$ and $v$ directly from the definition. 

\begin{figure}[t]
	\centering
    \includegraphics[scale=0.78]{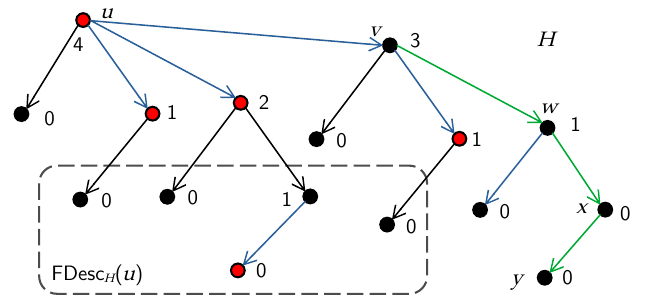}
    \caption{A \ba $H$ and the feedback descendants of its root. Feedback vertices are in red and demand arcs are in blue. $\alpha^*_g$ (where $g$ is $\alpha_H$ restricted to $V(F)$) is noted next to each vertex. $\alpha^*_g(v)=3$ due to its demand siblings. The green arcs are those that would be added during the inner loop of the algorithm when $v_{n-i}=v$ and $j=2$. The call to \textsc{Pack} would use $P=\{w,x,y\}$, add the arcs $(w,x)$ and $(x,y)$ and output $w$. Finally Step \ref{step:edge} would add the arc $(v, w)$.}
    \label{fig}
\end{figure}

\paragraph{The subroutine \textsc{Pack}.}

The following lemma describes a crucial subroutine for us. Informally, the subroutine creates a \pba of height $j$ by adding arcs between the provided vertices in one of its inputs ($P$) and then outputs the root of this \pba. Crucially, it is a very local algorithm: it only affects vertices from $P$. This allows us to repeatedly call it without undoing work it has already done in a previous call.
In order to describe this property we will need the following definition: suppose $Q$ and $Q'$ are both rooted forests, then define $\descendant_{Q'}^{Q}(v)$ as the set of vertices that are descendants of $v$ in $Q'$ but that have no parent in $Q$.
Note that all of these vertices except $v$ must have a parent in $Q'$.

\begin{lemma}[Packing lemma] \label{lem:pack}
	There is a polynomial-time algorithm \textsc{Pack} that:
	\begin{itemize}
		\item Takes as input a tuple $(Q, P, j)$ such that:
		\smallskip
		\begin{enumerate}
			\item $Q \subset T$ is a valid rooted forest.
			\item For every $w \in P$:
			\begin{enumerate}
				\item $Q[\descendant_Q(w)]$ is a \pba of height at most $j$.
				\item $w$ has no parent in $Q$.
			\end{enumerate}
			\item $\sum_{w \in P}{\beta_{Q}(w)} \geq 2^j$
		\end{enumerate}
		\medskip
		\item Outputs a tuple $(Q', v)$ such that:
		\smallskip
		\begin{enumerate}
			\item $Q' \subset T$ is a valid rooted forest and $Q'$ is a supergraph of $Q$.
			\item $Q'[\descendant_{Q'}(v)]$ is a \pba of height $j$.
			\item $v$ has no parent in $Q'$.
			\item If $Q'[\descendant_{Q'}(x)] \neq Q[\descendant_Q(x)]$ then $x \in \descendant_{Q'}^Q(v)$. Additionally $\descendant_{Q'}^Q(v) \subseteq P$. That is, all vertices affected by \textsc{Pack} become descendants of $v$ in $Q'$ and were from $P$.
			\item For all $s \in P \setminus \descendant_{Q'}^Q(v)$ and $t \in \descendant_{Q'}^Q(v)$ we have $\beta_Q(s) \leq \beta_Q(t)$. That is, the algorithm uses the vertices with largest height.
		\end{enumerate}
	\end{itemize}
\end{lemma}

\begin{proof}
	Initially let $Q_0 = Q$ and $P_0 = P$.
	For each $i \geq 0$ either:
	\begin{itemize}
		\item There exists $v \in P_i$ with $\beta_{Q_i}(v) = 2^j$. Then return $(Q_i, v)$.
		\item Or there is no such $v$. Then let $x, y$ be two vertices of largest $\beta$ with $x$ the stronger. More precisely, choose $x, y \in P_i$ such that $\beta_{Q_i}(x) = \beta_{Q_i}(y)$ and for all other $a, b \in P_i$ satisfying $\beta_{Q_i}(a) = \beta_{Q_i}(b)$ we have $\beta_{Q_i}(a) \leq \beta_{Q_i}(x)$. Let $Q_{i+1} = Q_i \cup \{(x, y)\}$ and $P_{i+1} = P_i \setminus \{y\}$ then repeat for the next $i$.
	\end{itemize}
    \paragraph{Correctness}
    We first show that the total value of $\beta$ across the vertices of $P_i$ is conserved. This is because, when we remove $y$ from $P_i$ after making $y$ a child of $x$, $\beta_{Q_i}(y)$ gets added to $\beta_{Q_{i+1}}(x)$. The formal argument follows. 
    Let $(x, y)$ be the arc present in $Q_{i+1}$ but not $Q_i$. Then $P_{i+1} = P_i \setminus \{y\}$.
    For every vertex $w \in P_{i+1}$ except $x$ we have $\beta_{Q_{i+1}}(w) = \beta_{Q_i}(w)$.
    Also, every vertex in $P_{i+1} \subseteq P$ has no parent in $Q$ and so, is not a feedback vertex. In particular, 
    \begin{align}
    	\beta_{Q_{i+1}}(x) &= 1 + \sum_{w \in \child_{Q_{i+1}}(x)}{\beta_{Q_{i+1}}(w)} \nonumber \\
    	&= \beta_{Q_i}(y) + 1 + \sum_{w \in \child_{Q_i}(x)}{\beta_{Q_i}(w)} \nonumber \\
    	&= \beta_{Q_i}(y) + \beta_{Q_i}(x). 
    	\label{eq:betasum}
    \end{align}
    Note that, since the descendants of every vertex $w \in P$ in $Q$ form a \pba, $\beta_Q(w)$ is a power of two and hence (\ref{eq:betasum}) shows that for all $i \geq 0$ and $w \in P_i$ we have $\beta_{Q_i}(w)$ is a power of two.
    Now
    \begin{flalign}
    	&&\sum_{w \in P_{i+1}}{\beta_{Q_{i+1}}(w)} &= \beta_{Q_{i+1}}(x) + \sum_{w \in P_{i+1} \setminus \{x\}}{\beta_{Q_i}(w)} && \nonumber \\ 
    	&& &=  \beta_{Q_i}(x) + \sum_{w \in P_i \setminus \{x\}}{\beta_{Q_i}(w)} && \nonumber \\ 
    	&& &= \sum_{w \in P_i}{\beta_{Q_i}(w)}
    	\label{eq:sumpreserved}
    \end{flalign}
    We can now argue that the algorithm always terminates. Let $n = |P|$. In each iteration one vertex is removed from $P_i$, so after $n - 1$ iterations, if the algorithm has not already returned, there is only one vertex remaining in $P_{n - 1}$, call it $x$. Since the total value of $\beta$ across the vertices of $P_i$ is conserved we know that $\sum_{w \in P_{n-1}}{\beta_{Q_{n-1}}(w)} = \beta_{Q_{n - 1}}(x) \geq 2^j$. Each vertex in $P_0$ is the root of a {\pba} in $Q_0$, after joining two such vertices by an arc we have a {\pba} of height one larger by definition. So for all $w \in P_i$, $\beta_{Q_{i+1}}(w) \leq 2\beta_{Q_i}(w)$ and both numbers are powers of two. Hence the algorithm will return before we have $\beta_{Q_i}(w) > 2^j$. Hence, $\beta_{Q_{n - 1}}(x) \leq 2^j$, implying that  $\beta_{Q_{n - 1}}(x)=2^j$ and the algorithm will return $(Q_n, x)$. 

    We now check that $Q'$ satisfies each of the conditions of the lemma:
    \begin{enumerate}
    	\item By construction, $Q'$ is a supergraph of $Q$ and a valid rooted forest (recall $Q$ is already valid). The only arcs in $Q'$ but not $Q$ are chosen to be from a stronger to a weaker vertex. Since at most one vertex without a parent is a feedback vertex (the root) and the new arcs are always between vertices of $P$ which all have no parent, $Q' \subset T$.
    	\item Each vertex in $P_0$ is the root of a {\pba} in $Q_0$, after joining two such vertices by an arc we have a {\pba} of height one larger by definition. So, for each $i$, the descendants of every vertex in $P_i$ form a {\pba} in $Q_i$, in particular $v$ in $Q'$. The height of this {\pba} rooted at $v$ is $j$ since $\beta_{Q'}(v) = 2^j$.
    	\item Every vertex in $P_0$ has no parent in $Q = Q_0$. Any vertex that is given a parent in $Q_{i+1}$ is removed from $P_{i+1}$ so the invariant that each vertex in $P_i$ has no parent in $Q_i$ is maintained. Since $v$ is chosen from $P_i$ it has no parent in $Q'$.
    	\item Every arc is added between vertices of $P$ that have the largest $\beta$ in $P_i$. Since after this the new root has even larger $\beta$, it will either be returned in the next iteration or have a new arc added in future. Hence the endpoints of every new arc end up forming a {\pba} rooted at $v$ and become descendants of $v$ in $Q'$ and since they are in $P$ they had no parent in $Q$. These are the only vertices whose descendants change since $v$ has no parent in $Q'$.
    	\item This is ensured by our choice of $x, y$.
    \end{enumerate}
\end{proof}

We will also need the following corollary of Lemma \ref{lem:pack}.
\begin{corollary} \label{cor:packwithba}
	Lemma \ref{lem:pack} holds if every occurrence of {\pba} is replaced by {\ba}.
\end{corollary}
Informally this says that the same subroutine can be used to pack {\ba}s: by replacing {\pba}s with {\ba}s in the input we get a \ba in the output.

We can now present our main result.

\begin{theorem} \label{thm:demandtfxp}
	An instance of {\sc Demand-TF}, $(T, \cS)$, can be solved in time $n^{\bigoh(k)}$ where $k$ is the feedback arc set number of $T$.
\end{theorem}

We first describe the algorithm claimed in the above statement, following which we give a proof of correctness. 

In our algorithm, we first guess an injective $p : V(F) \rightarrow V(T) \cup \{\bot\}$ representing the guessed parent of each feedback vertex. Then, for each $v \in V(F)$, we add the arc $(p(v), v)$ to $\cS$ (unless $v \in \head(\cS)$ or $p(v) = \bot$) to ensure this guess is honored by the algorithm and justify our assumption that every feedback vertex has a parent. We then guess $g : V(F) \rightarrow [\log(n)]$ representing the guessed heights of the feedback vertices and calculate $\alpha^*_g$. Now we can sanity check $g$: if $(u, v) \in \cS$ we check that  $\alpha^*_g(u) > \alpha^*_g(v)$ and if $(u, v), (u,w) \in \cS$ we check that $\alpha^*_g(v) \neq \alpha^*_g(w)$. Finally, if $p(v) = \bot$, we check that $g(v)=\log(n)$ since this is guessing that $v$ has no parent, i.e., it is the root.
Following this, we call Algorithm \ref{alg} once with each possible combination of the guesses of $g$ and value of $\Scal$ given by our guess of $p$. We reject the input if every invocation of Algorithm \ref{alg} fails to output a solution. 

\begin{algorithm}[h]
	\caption{}  \label{alg}
	$Q_{0,0} \leftarrow (V(T), \mathcal{S})$\;
	\For {$0 \leq i < n$}{
		\For {$0 \leq j < \alpha^*_g(v_{n-i})$}{
			\eIf {there exists $y \in \child_{Q_{i, 0}}(v_{n-i})$ with $\beta_{Q_{i, 0}}(y) = 2^j$ \label{step:if}}{
				$Q_{i, j+1} \leftarrow Q_{i, j}$\;
			}{
				Let $P_{i,j}$ be the set of vertices, $w$, that are weaker than $v_{n-i}$, have no parent in $Q_{i,j}$, and have $\beta_{Q_{i,j}}(w) \leq 2^j$\;
				\If {$\sum_{w \in P_{i, j}}{\beta_{Q_{i,j}}(w)} < 2^j$}{
					Reject\; \label{step:reject}
				}
				$(\widehat{Q}_{i, j+1}, w_{i, j}) \leftarrow \textsc{Pack}(Q_{i, j}, P_{i, j}, j)$\; \label{step:pack}
				$Q_{i, j+1} \leftarrow \widehat{Q}_{i, j+1} \cup \{(v_{n-i}, w_{i,j})\}$\; \label{step:edge}
			}
		}
		$Q_{i+1,0} \leftarrow Q_{i, \alpha^*_g(v_{n-i})}$\; \label{step:heightensured}
	}
	Let $P^*$ be the set of vertices without parents in $Q_{n,0}$\;
	\If{$\sum_{z \in P^*}{\beta_{Q_{n,0}}(z)} < n$}{
		Reject\; \label{step:rejectend}
	}
	$(Q^*, v^*) \leftarrow \textsc{Pack}(Q_{n,0}, P^*, \log(n))$\; \label{step:finalpack}
	\Return $Q^*$\;
\end{algorithm}

The outer loop iterates over each vertex from weakest to strongest. By Proposition \ref{prop:altba}, $v_{n-i}$ needs to have a child of height $j$ for each $j$ that the second loop considers. Step \ref{step:if} checks if such a child already exists (this happens when it is a demand child). Otherwise we check if there are enough vertices that could become descendants of $v_{n-i}$ and then call \textsc{Pack} to create such a child (see Figure \ref{fig}). After applying this process to every vertex we will have a number of {\ba}s that we can then pack into an \sba at Step \ref{step:finalpack}.

Recall that $\alpha^*_g$ can be calculated in polynomial time. Also $\alpha^*_g \leq \log(n)$ so both loops run at most $n$ times. Recall that each value of $\beta$ can be calculated in polynomial time. There are at most $n$ children of any vertex so the existence of $y$ can be checked in polynomial time. Each $P_{i,j}$ can be calculated in polynomial time by checking whether each vertex satisfies the conditions on vertices of $P_{i,j}$. Finally \textsc{Pack} is a polynomial-time subroutine. So since there are $\log(n)^{\bigoh(k)}$ possible values for $g$ and $n^{\bigoh(k)}$ possible values for $p$ (and the resulting value of $\Scal$), the overall running time is $n^{\bigoh(k)}$. 

\begin{proof}[Proof of correctness]
We first consider the case where the algorithm outputs $Q^*$ and prove that $Q^*$ is a valid \sba and hence the algorithm has correctly identified a positive instance.

\begin{lemma} \label{lem:correctdescendants}
	For each $0 \leq i < n$, if $v_{n-i}$ has no parent in $Q_{n,0}$, then the descendants of $v_{n-i}$ in $Q_{n,0}$ form a \ba of height $\alpha^*_g(v_{n-i})$
\end{lemma}

\begin{proof}
	We will begin with the following inductive hypothesis.
	\begin{hypothesis}
		For all $i < i'$, for all $j < \alpha^*_g(v_{n-i})$, either:
		\begin{itemize}
			\item there exists $y \in \child_{Q_{i, j}}(v_{n-i})$ such that the descendants of $y$ in $Q_{i,j}$ form a \pba of height $j$, or
			\item $(Q_{i,j}, P_{i,j}, j)$ satisfy the conditions of Lemma \ref{lem:pack}.
		\end{itemize} 
	\end{hypothesis}
	
	When $i=0$ there are no vertices weaker than $v_{n-i} = v_n$ and hence $P_{0,j} = \emptyset$. So the first possibility in the hypothesis must occur for every $j$ since otherwise the algorithm would have rejected at Step \ref{step:reject}.
	
	Now we use this hypothesis to prove the following claim.
	\begin{cclaim} \label{cl:childrenheights}
		For all $i < i'$, for all $\ell \leq \alpha^*_g(v_{n-i'})$, if either $v_{n-i} \in \head(\mathcal{S})$ or $v_{n-i}$ has no parent in $Q_{i',\ell}$, then the descendants of $v_{n-i}$ in $Q_{i',\ell}$ form a \pba of height $\alpha^*_g(v_{n-i})$.
	\end{cclaim}
	\begin{proof}
		We aim to find, for all $j < \alpha^*_g(v_{n-i})$, a $y \in \child_{Q_{i',\ell}}(v_{n-i})$ such that the descendants of $y$ in $Q_{i',\ell}$ form a \pba of height $j$.
		
		When the first possibility in the hypothesis occurs, the algorithm checks this and sets $Q_{i,j+1} \leftarrow Q_{i,j}$. Otherwise we can apply Lemma \ref{lem:pack} and hence the descendants of $w_{i,j}$ in $\widehat{Q}_{i,j+1}$ form a \pba of height $j$. So choosing $y = w_{i,j}$ guarantees we have such a $y$ in $Q_{i,j+1}$.
		
		Since $v_{n-i}$ either has a parent already in $Q_{i,j}$ or has no parent still in $Q_{i',\ell}$ it is not in $\descendant_{Q_{i,j'+1}}^{Q_{i,j'}}(w_{i,j'})$ for any $j \leq j' \leq \ell$ so its descendants are unchanged by \textsc{Pack} up to $Q_{i',\ell}$. Any other arcs are added below vertices stronger than $v_{n-i}$ so only affect the descendants of $v_{n-i}$ if they are below a feedback vertex, say $z$. In this case the new descendants are chosen to be of height less than $\alpha_g^*(z)$ so they have no effect on whether a \pba is formed below $y$ or its height. So the descendants of $y$ in $Q_{i', \ell}$ also form a \pba of height $j$. In $Q_{0,0}$, $v_{n-i}$ had no children of height greater than or equal to $\alpha_g^*(v_{n-i})$ since all such children would be demand children and this would contradict Definition \ref{def:alphastarwithg}. No arcs adjacent to $v_{n-i}$ are added before $Q_{i,0}$ so, in $Q_{i', \ell}$, $v_{n-i}$ has children whose descendants form a \pba of each height up to $\alpha^*_g(v_{n-i})$ and no others: this is exactly the conditions for Corollary \ref{cor:altpba}. 
	\end{proof}
	
	We can now extend the hypothesis to $i=i'$. We first show that $Q_{i',j}$ is a valid rooted forest and a subgraph of $T$ for all $j < \alpha^*_g(v_{n-i'})$. Initially $Q_{i',0} = Q_{i'-1,\alpha^*_g(v_{n-i'+1})}$ so is a valid rooted forest and a subgraph of $T$ by assumption. We proceed by induction on $j$. Either $Q_{i',j+1} = Q_{i',j}$ (and we are done) or it is just $Q_{i',j}$ with the additional arcs added at steps \ref{step:pack} and \ref{step:edge}. Since \ref{step:pack} is just a call to \textsc{Pack}, $\widehat{Q}_{i', j+1}$ is a valid rooted forest and a subgraph of $T$ by Lemma \ref{lem:pack} assuming that $(Q_{i',j}, P_{i',j}, j)$ satisfies the premises of Lemma \ref{lem:pack}. Also by Lemma \ref{lem:pack} $w_{i',j+1}$ has no parent in $\widehat{Q}_{i', j+1}$ and is in $P_{i', j+1}$ and hence weaker than $v_{n-i'}$ and not a feedback vertex. So $Q_{i', j+1}$ is also a valid rooted forest and a subgraph of $T$.
	
	It remains to show that, when the first possibility of the hypothesis does not occur, the conditions of Lemma \ref{lem:pack} on $P_{i',j}$ are satisfied. The way $P_{i',j}$ is chosen and that we have not rejected at Step \ref{step:reject} show all these conditions except that, for each $z \in P_{i',j}$, $Q_{i',j+1}[\descendant_{Q{i',j+1}}(z)]$ is a \pba. $z$ clearly has no parent in $Q_{i',j}$ and for each $z$ there is some $i < i'$ with $z = v_{n-i}$ since $z$ is weaker than $v_{n-i'}$ so this final condition is shown by Claim \ref{cl:childrenheights}.
	
	Therefore, by induction on $i'$ we have the hypothesis for all $i < n$ and hence also Claim \ref{cl:childrenheights} for all $i < n$.
	It remains to prove that the descendants of $v_{n-i}$ in $Q_{n,0}$ form a \ba and not just a \pba. Suppose for a contradiction this was not the case, then there must be some $i$ such that, for every $u \in \descendant_{Q_{n,0}}(v_{n-i})$ the descendants of $u$ do form a \ba but the descendants of $v_{n-i}$ itself do not. In particular the descendants of every child of $v_{n-i}$ form a \ba and hence the only way the descendants of $v_{n-i}$ do not form a \ba is for $v_{n-i}$ to be a feedback vertex and ``missing'' a child, that is there exists a $j < \alpha^*_g(v_{n-i})$ such that there is no child of $v_{n-i}$ whose descendants form a \ba of height $j$. However, in the proof of Claim \ref{cl:childrenheights} we find exactly such a child which is a contradiction. This extends Claim \ref{cl:childrenheights} to {\ba}s for $Q_{n,0}$, proving the Lemma.
\end{proof}

Hence $(Q_{n,0}, P^*, \log(n))$ satisfies the premises of Lemma \ref{lem:pack} and therefore, by Corollary \ref{cor:packwithba}, the descendants of $v^*$ in $Q^*$ form a \ba of height $\log(n)$ which is an \sba. So when the algorithm outputs a graph it is a valid \sba.

It remains to show the converse. Suppose for a contradiction that the algorithm rejects for every guess of $g$ but we have a positive instance: that is there exists a valid \sba. By Lemma \ref{lem:alwaysweaklycompact} there exists a valid, weakly compact \sba $H \subset T$. 

\begin{lemma} \label{lem:converseexchange}
	Given a valid, weakly-compact \sba $H$, for all $i, j$ there exists a valid \sba $H_{i,j}$ that is compact with respect to $\alpha_H$ restricted to $V(F)$, such that $Q_{i,j} \subset H_{i,j}$.
\end{lemma}

\begin{figure*}[t]
	\begin{center}
		\includegraphics[scale=0.73]{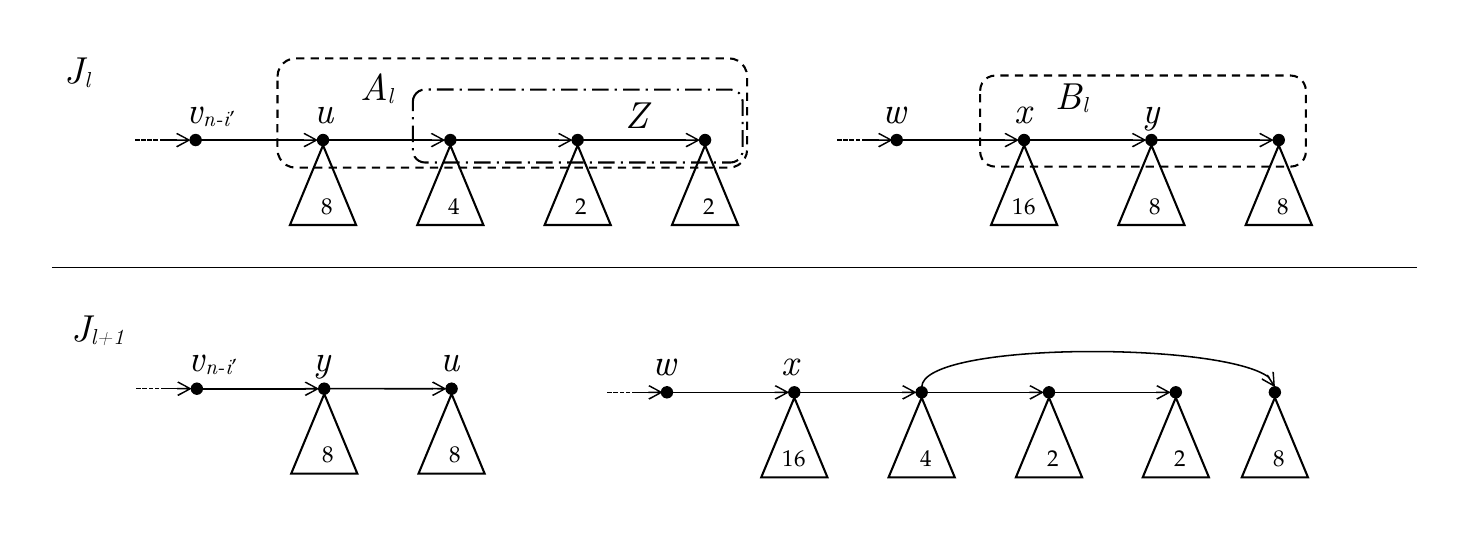}
		\caption{An example of the operation of Lemma \ref{lem:converseexchange}. The numbers represent $\beta_K$. Note that in this case $y$ beats $u$ so has taken the place of $u$ as the child of $v_{n-i}$ of size 16. Both $v_{n-i}$ and $w$ have many other children not pictured here. Note also that it may be the case that $v_{n-i} = w$.}
	\end{center}
\end{figure*}

\begin{proof}
	Let $g$ be $\alpha_H$ restricted to $V(F)$. Clearly $Q_{0,0} \subset H$ since $H$ is a valid {\sba} and so, contains every demand arc. Also $Q_{0,j}=Q_{0,0}$ for each  $0 \leq j < \alpha^*(v_{n-i})$ (see the proof of Lemma \ref{lem:correctdescendants}) so let $H_{0,j} = H$.
	Now assume that $Q_{i,j} \subset H_{i,j}$ for $0 \leq i < i'$ and $0 \leq j \leq \alpha^*(v_{n-i})$. We want to extend this to $i=i'$. Initially, $Q_{i', 0} = Q_{i'-1,\alpha^*_g(v_{n-i'+1})}$ so let $H_{i',0} = H_{i'-1,\alpha^*_g(v_{n-i'+1})}$.  Now we can assume that $Q_{i',j} \subset H_{i',j}$ and aim to prove that $Q_{i',j+1} \subset H_{i',j+1}$. Additionally if $v_{n-i'}$ is not a demand vertex then $\alpha^*_g(v_{n-i})=0$ so we are already done. Let $u$ be the unique vertex of height $j$ from $\child_{H_{i',j}}(v_{n-i'})$. If $(v_{n-i'}, u) \in \mathcal{S}$ then $Q_{i',j+1} =Q_{i',j}$ so we can set $H_{i',j+1}=H_{i',j}$ and we are done. Otherwise, let $R = \descendant_{Q_{i',j+1}}^{Q_{i',j}}(w_{i',j})$ (these are the vertices given parents by $\textsc{Pack}$), initialise $J_0=H_{i',j}$, and repeat the following for $\ell \geq 0$:
	\begin{enumerate}
		\item Let $u$ be the unique vertex of height $j$ from $\child_{J_\ell}(v_{n-i'})$ and $A_\ell = \descendant_{J_\ell}^{Q_{i',j}}(u) \setminus \fdesc_{J_\ell}(u)$.
		\item If $A_\ell = R$ then set $H_{i',j+1} = J_\ell$ and we are done.
		\item Otherwise, pick a vertex $y \in R \setminus A_\ell$ maximizing $\beta_{Q_{i',j}}(y)$. If $y = \troot(J_\ell)$ let $x = y$. Otherwise, let $w$ be the last vertex on the path from $\troot(J_\ell)$ to $y$ in $J_\ell$ that is in $\head(\mathcal{S}) \cup \{\troot(J_\ell)\}$.
		Let $x$ be the vertex after $w$ on this path.
		\item Let $B_\ell = \descendant_{J_\ell}^{Q_{i',j}}(x) \setminus \fdesc_{J_\ell}(x)$. Let $K$ be the rooted forest obtained from $J_\ell$ by deleting each in-edge to a vertex in $A_\ell \cup B_\ell$ (this includes $(v_{n-i'}, u)$ and $(w, x)$ if $y$ is not the root of $J_\ell$).
		\item Find a $Z \subset A_\ell \setminus R$ such that $\sum_{z \in Z}{\beta_{Q_{i',j}}(z)} = \beta_{Q_{i',j}}(y)$.
		\item Let $(K', a^*) \leftarrow \textsc{Pack}(K, (A_\ell \setminus Z) \cup \{y\}, \alpha_{J_\ell}(u))$\\ and  $(K'', b^*) \leftarrow \textsc{Pack}(K', (B_\ell \setminus \{y\}) \cup Z, \alpha_{J_\ell}(x))$.
		\item Finally, let $J_{\ell+1} = K'' \cup \{(v_{n-i'},a^*)\}$ $(\cup \{(w,b^*)\} \text{ if } y \neq \troot(J_\ell))$.
	\end{enumerate}
	Some explanation: $A_\ell$ are the vertices that $J_\ell$ has used to make the subtree of height $j$ below $v_{n-i'}$. We can think of these as the ``decisions'' that $J_\ell$ has made. If they agree with the decisions that $Q_{i',j+1}$ has made ($R$) then we are done. If not we find a vertex that $Q_{i',j+1}$ chose but $J_\ell$ did not ($y$) and look at where it is in $J_\ell$. $B_\ell$ is the decisions that $J_\ell$ made around $y$. By swapping the set that $y$ is in (and preserve sum of betas using $Z$) we can calculate a new $J_{\ell+1}$ where $y$ is in $A_{\ell+1}$.
	
	We now justify that this process does indeed produce a valid \sba $H_{i',j+1}$ as required.
	Clearly every vertex in $A_\ell$ is weaker than $v_{n-i'}$ and, since $y \in R \subseteq P_{i',j}$, $y$ is weaker than $v_{n-i'}$ so $(v_{n-i'}, a^*)$ is in $T$.
	
	If $y \neq \troot(J_\ell)$ then we need to show that $(w, b^*)$ is in $T$. Since $w$ is the \emph{last} vertex from $\head(\mathcal{S})$ and $x$ is after it, $x \notin \head(\mathcal{S})$, that is $(w, x) \notin \mathcal{S}$. Suppose for a contradiction that $w$ is strictly weaker than $v_{n-i'}$. The conditions on $w$ ensure Claim \ref{cl:childrenheights} applies: the descendants of $w$ in $Q_{i',j}$ form a \pba of height $\alpha^*_g(w)$. Therefore, since $Q_{i',j} \subset J_\ell$, $\descendant_{J_\ell}(w) \setminus \fdesc_{J_\ell}(w) = \descendant_{Q_{i',j}}(w) \setminus \fdesc_{Q_{i',j}}(w)$. But $y \in \descendant_{J_\ell}(w) \setminus \fdesc_{J_\ell}(w)$ whereas $y \notin \descendant_{Q_{i',j}}(w) \setminus \fdesc_{Q_{i',j}}(w)$ since $y$ has no parent in $Q_{i',j}$. 
	So $w$ is at least as strong as $v_{n-i'}$ (they may be the same vertex).
	Every vertex in $A_\ell$ (and hence $Z$) is weaker than $u$ since all feedback vertices have parents and any stronger vertices which are descendants of $u$ are excluded by $\fdesc_{J_\ell}(u)$. Since $u$ is weaker than $v_{n-i'}$ and additionally not a feedback vertex, the edge $(w, b^*)$ is in $T$.
	
	We need to check that we only make valid calls to \textsc{Pack}. 
	By our choice of $Z$ we have \begin{equation}
		\sum_{a \in (A_\ell \setminus Z) \cup \{y\}}{\beta_{K}(a)} = \sum_{a \in A_\ell}{\beta_{K}(a)} = \alpha_{J_\ell}(u) \label{eq:betaeq}
	\end{equation}
	and the same for $x$.
	Both sides of the last equality are simply counting the vertices that are descendants of $u$ in $H$: the left side uses $g$ to count descendants of feedback vertices and explicitly counts the others. By the assumption of compactness with respect to $g$ in the inductive hypothesis this agrees with the right hand side.
	The descendants of $u$ and $x$ in $J_\ell$ form a \pba since initially $J_\ell = H_{i',j}$ which is an \sba and the only changes are the result of \textsc{Pack}.
	
	It remains to show that such a $Z$ exists for each $\ell$.
	We have $A_\ell \setminus R \subset P_{i',j} \setminus R$ so $\beta_{Q_{i',j}}(z) \leq \beta_{Q_{i',j}}(y)$ for all $z \in A_\ell \setminus R$ by Lemma \ref{lem:pack}. Finally the equality 
	$$\sum_{z \in A_\ell}{\beta_{Q_{i',j}}(z)} = \sum_{z \in R}{\beta_{Q_{i',j}}(z)} = 2^j$$
	is conserved for all $\ell$ by the choice of $Z$. Hence
	$$\sum_{z \in A_\ell \setminus R}{\beta_{Q_{i',j}}(z)} = \sum_{z \in R \setminus A_\ell}{\beta_{Q_{i',j}}(z)} \geq \beta_{Q_{i',j}}(y)$$
	because $y$ is one of the elements of the second sum. Therefore there exists a $Z$ for each $\ell$.
	
	Since $Z \subset A_\ell \setminus R$ and $y \in R \setminus A_\ell$, $|A_{\ell} \cap R|$ increases by one each iteration, so after at most $|R|$ iterations the loop will terminate. At this point $A_\ell = R$, that is,
	the descendants of the child of $v_{n-i'}$ of height $j$ (including the child itself) are the same in both $Q_{i,j+1}$ and $H_{j+1}$. 
	Also every vertex in $\head(\mathcal{S})$ has a parent in $Q_{i',j}$ so their descendants are unchanged by \textsc{Pack} and since $\alpha_{J_{\ell+1}}(a^*) = \alpha_{J_\ell}(u)$, no heights are changed there either. Hence their heights agree between $H_{i',j}$ and $H_{i',j+1}$ and, since $H_{i',j}$ is compact with respect to $g$, $H_{i',j+1}$ is compact with respect to $g$ too. So by induction on $j$ we have a valid {\sba} $H_{i',j}$ which is compact with respect to $g$ and a supergraph of $Q_{i',j}$ for each $j$. This completes the proof of the lemma.
\end{proof}

Consider the run where the algorithm guesses $g$ as $\alpha_H$ restricted to $V(F)$. First, suppose the algorithm rejects at Step \ref{step:reject}. $H_{i, j}$ is compact with respect to $g$ so $\alpha_{H_{i,j}}(v_{n-i}) \geq \alpha^*_g(v_{n-i})$, by Observation \ref{obs:atleastalphastar}. Hence there exists $x \in \child_{H_{i,j}}(v_{n-i})$ with $\alpha_{H_{i,j}}(x) = j < \alpha^*_g(v_{n-i})$ and $x$ is not a demand vertex since the algorithm never rejects in this case. Now $\sum_{y \in \descendant_{H_{i,j}}^{Q_{i,j}}(x)}{\beta_{Q_{i, j}}(y)} \geq 2^j$ by a similar argument as for (\ref{eq:betaeq}). This time the feedback descendants are included in the sum and hence are double counted leading to an inequality. All vertices $y \in \descendant_{H_{i,j}}^{Q_{i,j}}(x)$ are weaker than $x$ and hence $v_{n-i}$ since $x$ is not a feedback vertex. Also $2^j \geq \beta_{H_{i,j}}(y) \geq \beta_{Q_{i,j}}(y)$ since $Q_{i,j} \subset H_{i,j}$. Finally they are chosen to not have a parent in $Q_{i,j}$ so they must be in $P_{i,j}$. But then we have $$2^j \leq \sum_{y \in \descendant_{H_{i,j}}^{Q_{i,j}}(x)}{\beta_{Q_{i, j}}(y)} \leq \sum_{w \in P_{i,j}}{\beta_{Q_{i, j}}(w)}$$ which contradicts the rejection of the algorithm.

Finally, suppose the algorithm rejects at Step \ref{step:rejectend}. Then consider $H_{n,0} = H_{n-1,\alpha_g^*(v_1)} \supset Q_{n-1,\alpha_g^*(v_1)} = Q_{n,0}$. Since $H_{n,0}$ is an \sba, clearly $$\sum_{z \in P^*}{\beta_{Q_{n,0}}(z)} \geq 2^{\alpha_{H_{n,0}}(\troot(H_{n,0}))} = n$$ contradicting the rejection.
\end{proof}

In the following sections, we describe how our algorithm can be extended or modified to obtain the further algorithmic results outlined in the Introduction.

\subsection{FPT Algorithm for \textsc{Demand-TF} When Upsets Are Demanded}
Suppose that $F \subseteq \mathcal{S}$. 
A closer inspection of our algorithm shows that where we assumed that every feedback vertex except the root has a demand parent, a weaker assumption is sufficient, specifically, the following property.
\begin{property} \label{prop:noinfeedback}
	For all $v \notin \head(\mathcal{S})$ there is no $(u, v) \in F$.
\end{property}
That is, every vertex without a demand parent has no in-feedback arc. Equivalently $\head(F) \subseteq \head(\cS)$. Clearly this is implied by $F \subseteq \cS$.
This means that any vertex that is stronger than $v$, and only these vertices, can be used as its parent. Furthermore, even if $v$ is a feedback vertex the descendants of $v$ are weaker than its ancestors, at least until further feedback vertices. Effectively $v$ does not act as a feedback vertex (if there is a feedback edge $(v, w)$ we will handle this when talking about $w$) so we replace $V(F)$ with $\head(F)$ throughout the algorithm: these are the remaining feedback vertices that still behave as such. Whenever the analysis refers to a vertex not being a feedback vertex because it has no demand parent we now simply use Property \ref{prop:noinfeedback}: although it may be a feedback vertex, it behaves as a non-feedback vertex since it has no incoming feedback edges.

The only changes required to our algorithm are therefore the removal of the guess of the parents of feedback vertices (eliminating the overhead of $n^{\bigoh(k)}$) and a slightly smaller set of feedback vertices used in the domain of $g$ and the definitions of $\alpha^*$ and $\beta$. Since the guess of the parents of the feedback vertices has been removed the runtime is dominated by the guess of $g$: there are $(\log n)^{\bigoh(k)} = (k \log k)^{\bigoh(k)}n^{\bigoh(1)}$ possibilities for $g$ so the overall run time is FPT in $k$.

\subsection{XP Algorithm for \textsc{Demand-TF} With Specified Rounds for Demands}
The first key observation is that the round of a match is exactly the height of that the losing player of the match takes in the solution (numbering rounds from zero). That is, if a demand match $(u, v)$ occurs at round $i$ in the tournament represented by an \sba $H$ then $\alpha_H(v) = i$. So if every demand match has a specified round this is equivalent to specifying the heights of every vertex in $\head(\cS)$. Theorem \ref{thm:demandtfxp} guesses the heights of vertices in $V(F)$ and calculates the heights of all other vertices in $\head(\cS)$. So by extending the domain of $g$ to $\head(\cS) \cup V(F)$ and ensuring it agrees with the required heights the algorithm will ensure that the \sba it outputs will satisfy our additional constraints. Note that when we extend the domain of $g$ in this way, the guess is still made only for the vertices in $V(F)$ as the value of $g$ for the vertices in $\head(\Scal)$ is part of the input. Thus, the running time of the new algorithm remains the same as that of Theorem~\ref{thm:demandtfxp}. 

\section{Future Work}
\begin{enumerate}
    \item Theorem~\ref{thm:demandtfxp} shows the viability of a systematic study of the parameterized complexity of {\demandTF} with respect to other parameters studied for TF, such as feedback {\em vertex} set.  
	\item We also propose a relaxation of our model where, if, there is no seeding that enables every demand match to take place, then the goal is to compute a seeding that makes the {\em maximum} number of demand matches take place. How efficiently could one do this? Our work implies a $2^dn^{\bigoh(k)}$-time algorithm for this problem, where $d$ is the total number of demands and $k$ is the feedback arc set number -- simply guess the set of satisfied demands and invoke Theorem~\ref{thm:demandtfxp}. A natural follow-up question is whether one can remove the exponential dependence on $d$ and have an XP algorithm parameterized by $k$ alone?
	Resolution of this question would require additional ideas to this paper. For instance, we can no longer assume that every vertex has at most one demand in-neighbor. Efficient approximation algorithms for this problem are also an interesting research direction. In terms of exact-exponential-time algorithms, it is easy to see that the algorithm of Theorem \ref{thm:exactExponentialTimeAlgorithm} already extends naturally to this variant. 
	\item Could one extend our results for {\demandTF} to the probabilistic model? Here, the tournament designer has two natural objectives -- maximize the probability that all demand matches are played or maximize the expected number of satisfied demands.
	\item The edge-constrained version of {\sc Subgraph Isomorphism} we define in this paper is a  problem of independent interest. Techniques such as color coding~\cite{AlonYZ95} could give FPT algorithms for this problem parameterized by the size of some simple pattern graphs. However, the behavior of this problem with respect to various structural parameterizations of graphs is less clear and we leave this as a research direction of broad interest to the algorithms community. 
\end{enumerate}

\section{Acknowledgements}
Sushmita Gupta acknowledges support from SERB's MATRICS Grant (MTR/2021/000869) and SUPRA Grant (SPR/2021/000860).
Ramanujan Sridharan acknowledges support by the Engineering and Physical Sciences Research Council (grant numbers EP/V007793/1 and EP/V044621/1). 

\clearpage
\bibliography{aaai24}

\end{document}